\documentclass{lmcs}
\pdfoutput=1

\usepackage{lastpage}
\lmcsdoi{15}{4}{4}
\lmcsheading{}{\pageref{LastPage}}{}{}%
{Dec.~08,~2016}{Oct.~29,~2019}{}

\usepackage{hyperref}
\usepackage{enumerate,bbm,tikz,cite}
\input{cherrypick}
\usepackage{ccs} 
\usetikzlibrary{arrows}

\usepackage{Regalg}

\keywords{Regular algebra; \texorpdfstring{$\omega$}{omega}-continuous algebra; iteration theories}

\ACMCCS{\printccsdesc}

\ccsdesc[500]{Theory of computation~Algebraic semantics}
\ccsdesc[300]{Theory of computation~Algebraic language theory}
\ccsdesc[300]{Theory of computation~Tree languages}
\ccsdesc[300]{Theory of computation~Categorical semantics}
\ccsdesc[100]{Theory of computation~Regular languages}

\begin{document}

\title[On Free \texorpdfstring{$\omega$}{omega}-Continuous and Regular Ordered Algebras]{On Free \texorpdfstring{$\omega$}{omega}-Continuous and Regular Ordered Algebras}

\author[Z.~\'Esik]{Zolt\'an \'Esik\rsuper{a}}
\address{\lsuper{a}Institute of Informatics, University of Szeged, P.O.~Box 652, H-6701 Szeged, Hungary}

\author[D.~Kozen]{Dexter Kozen\rsuper{b}}
\address{\lsuper{b}Computer Science Department, Cornell University, Ithaca, NY 14853--7501, USA}

\begin{abstract}
\noindent
We study varieties of certain ordered $\Sigma$-algebras with restricted completeness and continuity properties. We give a general characterization of their free algebras in terms of submonads of the monad of $\Sigma$-coterms. Varieties of this form are called \emph{quasi-regular}. For example, we show that if $E$ is a set of inequalities between finite $\Sigma$-terms, and if $\cVo$ and $\cVr$ denote the varieties of all $\omega$-continuous ordered $\Sigma$-algebras and regular ordered $\Sigma$-algebras satisfying $E$, respectively, then the free $\cVr$-algebra $\FrX$ on generators $X$ is the subalgebra of the corresponding free $\cVo$-algebra $\FoX$ determined by those elements of $\FoX$ denoted by the regular $\Sigma$-coterms. This is a special case of a more general construction that applies to any quasi-regular family. Examples include the *-continuous Kleene algebras, context-free languages, $\omega$-continuous semirings and $\omega$-continuous idempotent semirings, OI-macro languages, and iteration theories.
\end{abstract}

\maketitle

\section{Introduction}

We are concerned with varieties of certain ordered $\Sigma$-algebras with restricted completeness and continuity properties. There are many examples of such varieties; for example, the star-continuous Kleene algebras satisfy the property that any regular set of elements (a set definable by a regular expression) has a supremum. That is, if $K$ is a star-continuous Kleene algebra, $R$ is the canonical interpretation of regular expressions over a finite alphabet $A$ as sets of strings over $A$, and $I:A\to K$ is an interpretation in $K$, then $\sup\set{I(x)}{x\in R(e)}$ exists for any regular expression $e$. Star-continuity is axiomatized by a special case of this property, namely that $ab\star c = \sup_{n\ge 0} ab^{n}c$ for any $a,b,c\in K$. This axiom says that $b\star$ is the supremum of the $b^n$ and that multiplication is continuous with respect to such suprema.

The completeness and continuity properties of this and similar examples are \emph{restricted} in the sense that not all suprema need exist, but only those that are definable by some syntactic mechanism, depending on the variety. Other examples involve context-free languages, $\omega$-continuous semirings and $\omega$-continuous idempotent semirings, OI-macro languages, and iteration theories. We describe some of these in more detail in \S\ref{sec:examples}.

In this paper give a general account of the free algebras of such varieties in terms of submonads of the monad of $\Sigma$-coterms. Varieties of this form are called \emph{quasi-regular}. The intent is to provide a uniform framework in which to understand their completeness and continuity properties from a concrete perspective. Additionally, we show how relations between such varieties are reflected in relations between their free models; for example, we show that if $E$ is a set of inequalities between finite $\Sigma$-terms, and if $\cVo$ and $\cVr$ denote the varieties of all $\omega$-continuous ordered $\Sigma$-algebras and regular ordered $\Sigma$-algebras satisfying $E$, respectively, then the free $\cVr$-algebra $\FrX$ on generators $X$ is the subalgebra of the corresponding free $\cVo$-algebra $\FoX$ determined by those elements of $\FoX$ denoted by the regular $\Sigma$-coterms. This is a special case of a more general construction that applies to any quasi-regular family.

\section{\texorpdfstring{$\omega$}{Omega}-Continuous Algebras}

Let $\Sigma$ be a ranked alphabet, which will be fixed throughout.
A $\Sigma$-algebra is called \emph{ordered} if $A$ is partially ordered
by a relation $\le$ with least element $\bot^A$ and the algebraic operations
are monotone with respect to $\le$; that is, if $f \in \Sigma_n$ and
$a_i,b_i\in A$ with $a_i\le b_i$ for $1\le i\le n$,
then $f^A(\seq a1n) \le f^A(\seq b1n)$.
A morphism $h:A\to B$ of ordered $\Sigma$-algebras is a strict monotone
map that commutes with the algebraic operations:
\begin{gather*}
a \le b \Rightarrow h(a) \le h(b)\qquad\ h(\bot^A) = \bot^B\qquad\ h(f^A(\seq a1n)) = f^B(h(a_1),\ldots,h(a_n))
\end{gather*}
for all $a,b,\seq a1n \in A$ and $f\in\Sigma_n$, $n\geq 0$. We denote this category by $\SAlg$.

For each set $X$, there is a free ordered $\Sigma$-algebra $\FT X$ freely generated by $X$.
The elements of $\FT X$ are represented by the finite partial $\Sigma$-terms over $X$;
here \emph{partial} means that some subterms may be missing, which is the same as
having the empty term $\bot^{\FT X}$ in that position.
When $A$ is an ordered $\Sigma$-algebra, any function $X\to A$
extends uniquely to a morphism $\FT X\to A$. Thus the functor $\FT:\Set\to\SAlg$ is left adjoint
to the forgetful functor in the other direction.

An ordered $\Sigma$-algebra $A$ is \emph{$\omega$-continuous}~\cite{Courcelle83,ADJ77,Guessarian81} if
it is $\omega$-complete and the operations are $\omega$-continuous.
That is, any countable directed set (or countable chain) $C$ has a supremum $\bigvee C$,
and when $n\geq 0$, $f\in\Sigma_n$, and $C_i$ is a nonempty countable
directed set for $1\le i\le n$, then
\begin{align*}
f^A(\seq{\bigvee C}1n) &= \bigvee \set{f^A (\seq x1n)}{x_i\in C_i,\ 1\le i\le n}.
\end{align*}
A morphism of $\omega$-continuous $\Sigma$-algebras is an $\omega$-continuous ordered $\Sigma$-algebra morphism.
The category of $\omega$-continuous $\Sigma$-algebras is denoted $\SAlgo$.

For each set $X$, there is a free $\omega$-continuous
$\Sigma$-algebra $\CT X$ freely generated by $X$. The elements of $\CT X$
are the \emph{partial $\Sigma$-coterms} over $X$ (finite or infinite partial terms).
Formally, these are partial functions $t:\omega\star\pfun\Sigma+X$ with domain $\dom t$ such that
\begin{itemize}
\item
$\dom t$ is prefix-closed, and
\item
if $\alpha\in\omega\star$, $i\in\omega$, and $\alpha i\in\dom t$, then $t(\alpha)\in\Sigma_n$ for some $n\ge 1$ and $0\le i<n$ (that is, $t(\alpha)\not\in X\cup\Sigma_0$).
\end{itemize}
The ordering is defined by:
\begin{align*}
s \le t\ &\Leftrightarrow\ \forall\alpha\ \alpha\in\dom s \Rightarrow (\alpha\in\dom t \wedge s(\alpha)=t(\alpha));
\end{align*}
that is, $s$ can be obtained from $t$ by erasing some subterms.
The minimal element $\bot^{\CT X}$ is the unique such function with domain $\emptyset$.
The algebra $\FT X$ is the subalgebra of $\CT X$ consisting of those elements with finite domain.
As above, when $t\in \CT X$ and $A$ is an $\omega$-continuous $\Sigma$-algebra, any set function $X\to A$
extends uniquely to a morphism $\CT X\to A$ of $\omega$-continuous $\Sigma$-algebras.
Thus, as above, the functor $\CT:\Set\to\SAlgo$ is left adjoint
to the forgetful functor in the other direction.

Suppose that $X_\omega$ is a fixed countably infinite set. A set $E$ of formal inequalities
$t \sqsubseteq t'$ between terms in $\FT{X_\omega}$ determines
in the usual way a \emph{variety of
ordered $\Sigma$-algebras}, denoted $\cV$, and a \emph{variety of $\omega$-continuous
$\Sigma$-algebras}, devoted $\cVo$.
An ordered $\Sigma$-algebra $A$
belongs to $\cV$ if $h(s)\le h(t)$ in $A$ for any
morphism $h:\FT X_\omega\to A$ and $s\sqsubseteq t$ in $E$.
The class $\cVo$ contains all $\omega$-continuous
$\Sigma$-algebras in $\cV$.

It is known that free algebras exist in both $\cV$ and $\cVo$~\cite{Bloom76,Guessarian81}.
The free algebra $\FX$ in $\cV$ generated by a set $X$ is the algebra of terms $\FT X$
modulo the least congruence containing $E$ and the inequalities of $\FT X$.
This is also the $X$-generated ordered subalgebra of the
free $\omega$-continuous algebra $\FoX$ in $\cVo$.
Moreover, $\FoX$ is the completion of $\FX$ by $\omega$-ideals
(see \S\ref{sec:completion}).

We note that one of our examples of \S6, specifically Example~\ref{ex:iteration}, requires a typed or multisorted signature $\Sigma$, effectively restricting terms $\FT X$ and coterms $\CT X$ to well-typed submonads of the same. The exact requirement will be made clear in \S6. The cited results of~\cite{Bloom76,Guessarian81} are developed in the traditional untyped context for notational clarity, but it is straightforward to check that they extend to the typed case needed for this example.

\section{Free Extension}%
\label{sec:completion}

Every ordered algebra can be completed to an $\omega$-continuous algebra by the method of \emph{completion by $\omega$-ideals}. This construction constitutes a functor $\Io:\SAlg\to\SAlgo$ that is left adjoint to the forgetful functor in the other direction. This result is well known and can be proved by standard universal-algebraic argument~\cite{Bloom76,Guessarian81}, which we outline here.
We will later give a somewhat more general construction called \emph{free extension} that will allow us to apply this idea more widely.

For $C$ a subset of an ordered $\Sigma$-algebra $A$, define
\begin{align*}
\down C &= \set{c}{\exists a\in C\ c\le a}.
\end{align*}
A set $C$ is an \emph{order ideal} if it is nonempty, directed, and downward closed; that is,
if for any $a,b\in C$ there exists $c\in C$ such that $a,b\le c$, and $a\in C$ whenever
$a\le b$ and $b\in C$, i.e.~$C=\down C$. An order ideal $C\subs A$
is called an \emph{$\omega$-ideal} if it is countably generated; that is, there
is a countable directed set $C_0\subs C$ such that $C = \down{C_0}$.
The set $\Io A$ of all $\omega$-ideals of $A$ ordered by set inclusion $\subs$
is an $\omega$-complete poset with least element $\bot^{\Io A}=\{\bot^A\}$. The supremum
of an $\omega$-ideal $\AA$ of $\omega$-ideals is the union $\bigcup\AA$.
We can make $\Io A$ into an $\omega$-continuous algebra by defining
\begin{align}
f^{\Io A}(\seq C1n) &= \down{\set{f^A(\seq a1n)}{a_i\in C_i,\ 1\le i\le n}},\label{eq:distr}
\end{align}
the order ideal generated by the set $\set{f^A(\seq a1n)}{a_i\in C_i,\ 1\le i\le n}$. The action of $\Io$ on morphisms $h:A\to B$ is
\begin{align*}
& \Io h:\Io A\to\Io B && \Io h(C) = \down{\set{h(c)}{c\in C}}.
\end{align*}

The algebra $A$ can be embedded in $\Io A$ by
the ordered algebra morphism mapping $a\in A$ to the order
ideal $\down{\{a\}}$ generated by $a$. This is the unit of the adjunction.
The structure $\Io A$ is the \emph{free completion} of $A$ to an $\omega$-continuous ordered algebra~\cite{Bloom76,Guessarian81}.
In particular, $\CT X$ is isomorphic to $\Io(\FT X)$ for all $X$.

Ignoring the $\Sigma$-algebra operations for a moment, the essential structure of the functor $\Io$ as a completion procedure stems from a monad $(\Io,\bigcup,\down{\{-\}})$ on partial orders with $\bot$, where $\bigcup:\Io^2\to \Io$ and $\down{\{-\}}:\Id\to \Io$ are natural transformations satisfying the usual monad laws. Given $(X,\le,\bot^X)$, the monad produces the $\omega$-complete structure $(\Io X,\subs,\bot^{\Io X})$, where $\bot^{\Io X}=\{\bot^X\}$.

If $\Io$ is applied to an ordered $\Sigma$-algebra $A$, the operations of $\Sigma$ can be defined on $\Io A$ by~\eqref{eq:distr}, and the operations so defined are monotone and $\omega$-continuous: for $\seq\AA 1n\in \Io^2(A)$,
\begin{align*}
&f^{\Io A}(\seq{\bigcup\AA}1n) \\
&\quad= \down{\set{f^A(\seq a1n)}{a_i\in\bigcup\AA_i,\ 1\le i\le n}}\\
&\quad= \down{(\bigcup \set{\down{\set{f^A(\seq a1n)}{a_i\in C_i,\ 1\le i\le n}}}{C_i\in\AA_i,\ 1\le i\le n})}\\
&\quad= \bigcup\down{\set{ f^{\Io A}(\seq C1n)  }{C_i\in\AA_i,\ 1\le i\le n}}.
\end{align*}
The equation~\eqref{eq:distr} is essentially a distributive law~\cite{Beck69} that transforms a term over ideals to an ideal of terms, thereby allowing the algebraic operations of $\Sigma$ to commute with suprema.

\begin{exa}%
\label{ex:FTCT}
The free ordered algebra $\FT X$ and the free $\omega$-continuous ordered algebra $\CT X$ are Eilenberg-Moore algebras for the \emph{partial term monad} and \emph{partial coterm monad} for the signature $\Sigma$, respectively. These are monads on $\Set$ with components $\mu:\CT^2\to\CT$ and $\eta:\Id\to\CT$ (and similarly for $\FT$) defined as follows. For a set $X$ and $s\in\CT X$, define $\var s=\set{\alpha\in\dom s}{s(\alpha)\in X}$. For $t\in\CT^2 X$, $x\in X$, and $\alpha,\beta\in\omega\star$,
\begin{align*}
& \dom(\mu_X(t)) = \dom t \cup \set{\alpha\beta}{\alpha\in\dom t,\ t(\alpha)\in\var t,\ \beta\in\dom t(\alpha)} && \dom(\eta_X(x)) = \{\eps\}\\
& \mu_X(t)(\alpha) = t(\alpha),\ \alpha\in\dom t\setminus\var t && \eta_X(x)(\eps) = x\\
& \mu_X(t)(\alpha\beta) = t(\alpha)(\beta),\ \alpha\in\var t,\ \beta\in\dom t(\alpha)
\end{align*}
An $\omega$-ideal $C\subs\FT X$ is any nonempty downward-closed directed set. Any such set has a supremum in $\CT X$: $(\bigvee C)(\alpha) = t(\alpha)$ if $t\in C$ and $\alpha\in\dom t$, or undefined if no such $t$ exists.
%
The coterm $\bigvee C$ is well defined, as any two elements of $C$ must agree on the intersection of their domains, since $C$ is directed.
\end{exa}

More generally, let $K$ be a submonad of $\Io$. The monad operations are the same as in $\Io$.
An ordered $\Sigma$-algebra $A$ is said to be \emph{$K$-continuous} if it is $K$-complete and the algebraic operations are $K$-continuous.
That is, all ideals in $KA$ have suprema in $A$, and the algebraic operations preserve suprema.

One can perform the same completion construction with $K$ as with $\Io$, provided $KA$ forms a subalgebra of $\Io A$ under the operations~\eqref{eq:distr}; that is, the $\omega$-ideal
\begin{align}
f^{KA}(\seq C1n) &= \down{\set{f^A(\seq a1n)}{a_i\in C_i,\ 1\le i\le n}}\in KA\label{eq:continuous}
\end{align}
whenever $\seq C1n\in KA$.
The resulting ordered $\Sigma$-algebra $KA$ is called the \emph{free extension} of $A$ by the monad $K$. It is not $\omega$-continuous in general, but it is $K$-continuous; to wit, the supremum of an ideal $\AA\in K^2 A$ is $\bigcup\AA$, and the algebraic operations preserve them by the same proof as with $\Io$ above.

The free extension $KA$ of $A$ satisfies the following universal property analogous to~\cite[Theorem 3.3]{Bloom76} for $K=\Io$.

\begin{thm}%
\label{thm:univ}
Let $A,B$ be ordered $\Sigma$-algebras, $B$ $K$-continuous, and $h:A\to B$ an ordered $\Sigma$-algebra morphism. There is a unique $K$-continuous morphism $\hat h:KA\to B$ that agrees with $h$ on $A$ in the sense that $\hat h\circ\down{\{-\}} = h$.
\end{thm}
\begin{proof}
For $C\in KA$, define
\begin{align*}
\hat h(C) = \bigvee\set{h(a)}{a\in C}.
\end{align*}
The supremum on the right-hand side exists since $Kh(C)=\down{\set{h(a)}{a\in C}}\in KB$ and $B$ is $K$-complete. The map $\hat h$ agrees with $h$ on $A$ since $\hat h(\down{\{a\}}) = \bigvee\set{h(b)}{b\in\down{\{a\}}} = h(a)$. It is monotone, since if $C\subs D$, then
\begin{align*}
\hat h(C) = \bigvee\set{h(a)}{a\in C} \le \bigvee\set{h(a)}{a\in D} = \hat h(D).
\end{align*}
It is strict, since it agrees with $h$ on $A$ and $h$ is strict.

To show that $\hat h$ commutes with the algebraic operations,
\begin{align*}
\lefteqn{\hat h(f^{KA}(\seq C1n))}\quad\\
&= \hat h(\down{\set{f^A(\seq a1n)}{a_i\in C_i,\ 1\le i\le n}}) && \text{by definition of $f^{KA}$}\\
&= \bigvee\set{h(b)}{b\in\down{\set{f^A(\seq a1n)}{a_i\in C_i,\ 1\le i\le n}}} && \text{by definition of $\hat h$}\\
&= \bigvee\set{h(f^A(\seq a1n))}{a_i\in C_i,\ 1\le i\le n}\\
&= \bigvee\set{f^B(h(a_1),\ldots,h(a_n))}{a_i\in C_i,\ 1\le i\le n}\\
&= f^B(\bigvee\set{h(a_1)}{a_1\in C_1},\ldots,\bigvee\set{h(a_n)}{a_n\in C_n}) && \text{since $f^B$ is $K$-continuous}\\
&= f^B(\hat h(C_1),\ldots,\hat h(C_n)) && \text{by definition of $\hat h$.}
\end{align*}

To show that $\hat h$ is $K$-continuous, we wish to show that for any $\AA\in K^2 A$,
\begin{align*}
\hat h(\bigcup\AA) = \bigvee\set{\hat h(C)}{C\in\AA}.
\end{align*}
The supremum on the right-hand side exists since $K\hat h(\AA)=\down{\set{\hat h(C)}{C\in\AA}}\in KB$ and $B$ is $K$-complete. Then
\begin{align*}
\bigvee\set{\hat h(C)}{C\in\AA} &= \bigvee\set{\bigvee\set{h(a)}{a\in C}}{C\in\AA}
= \bigvee\set{h(a)}{a\in\bigcup\AA}
= \hat h(\bigcup\AA).
\end{align*}

Finally, we show that $\hat h$ is unique. Let $h':KA\to B$ be any other $K$-continuous morphism
that agrees with $h$ on $A$.
Let $C\in KA$ be arbitrary. Then $C = \bigcup\set{\down{\{a\}}}{a\in C}$ and
\begin{align*}
h'(C) &= h'(\bigcup\set{\down{\{a\}}}{a\in C})\\
&= \bigvee\set{h'(\down{\{a\}})}{a\in C} && \text{since $h'$ is $K$-continuous}\\
&= \bigvee\set{h(a)}{a\in C} && \text{since $h'$ agrees with $h$ on $A$}\\
&= \hat h(C).
\end{align*}
As $C$ was arbitrary, $h'=\hat h$.
\end{proof}

\section{Quasi-Regular Families}

In this section we study limited completeness and continuity conditions in which
not all suprema need exist, but only those of a certain syntactically restricted form.

For each set $X$, let $\Del X$ denote an ordered subalgebra of $\CT X$ containing $X$.
The set $\Del X$ is a set of partial coterms closed under the algebraic operations $\Sigma$.
Note that each finite term in $\FT X$ is in $\Del X$, thus $\FT X\subs\Del X\subs\CT X$.

We further assume for all $X$ and $Y$ and for all $\omega$-continuous morphisms $h:\CT X \to \CT Y$
with $\set{h(x)}{x\in X} \subs \Del Y$ that $\set{h(t)}{t\in\Del X} \subs \Del Y$.
A family of term algebras $\Del X$ for each set $X$ satisfying this property is called
a \emph{quasi-regular family}.

Under these conditions, the functor
$\Del$ (tacitly composed with the forgetful functor to $\Set$) determines a submonad $(\Del,\mu,\eta)$ of the coterm monad. The monad operations are the same, but with suitably restricted domain. Thus
$\eta_X:X\to\Del X$ makes a singleton term $x$ out of an element $x\in X$,
and $\mu_X:\Del^2X\to\Del X$ takes a coterm of coterms over
$X$ and collapses it to a coterm over $X$,
as described in Example~\ref{ex:FTCT}.

\begin{exa}
Two extremal examples of quasi-regular families are $\CT X$, the
partial coterms over $X$, and $\FT X$, the
partial terms over $X$. These are the maximum and minimum
quasi-regular families, respectively, for a given signature $\Sigma$.
The regular or algebraic trees~\cite{Courcelle83,Guessarian81},
or more generally, the set
of regular (or rational) trees of order $n$~\cite{Gallier81,Gallier85}
also form quasi-regular families.
Also, the union of the $n$-regular trees
over $X$ for all $n\geq 0$ is a quasi-regular family.
These examples are described more fully in \S\ref{sec:examples}.
\end{exa}

\begin{lem}
For any set $Y$, the set $\Del Y$ is uniquely determined by $\Del{X_\omega}$.
\end{lem}
\begin{proof}
Any injection $h:X\to Y$ lifts to a morphism $\Del h:\Del X\to\Del Y$ of $\omega$-continuous algebras. We claim that
\begin{align*}
\Del Y &= \set{\Del h(t)}{X\subs X_\omega,\ \text{$h:X\to Y$ is an injection},\ t\in\Del X}.
\end{align*}
The reverse inclusion holds by our assumption $\Del h(\Del X)\subs\Del Y$. For the forward inclusion, suppose
$s\in\Del Y$. Then $s\in\Del Y'$ for some finite or countable subset $Y'\subs Y$, as there are at most countably many
subterms of $s$. Let $X\subs X_\omega$ be of the same cardinality as $Y'$ and let $h:X\to Y'$ be a bijection.
Then $\Del h^{-1}(s)\in\Del X_\omega$, $h:X\to Y$ is an injection, and $s=\Del h(\Del h^{-1}(s))$.
\end{proof}

\begin{defi}
For coterms $s,t\in\CT X$, define $s\ll t$ if $s$ is finite, but agrees as a labeled tree with $t$ wherever it is defined;
that is, $s\in\FT X$ and $s\le t$. For $t\in\CT X$, the set $\set s{s\ll t}$ forms an $\omega$-ideal in $\FT X$.
\end{defi}

\begin{defi}%
\label{def:quasireg}
Let $A$ be an ordered $\Sigma$-algebra.
A set $B\subs A$ is called a \emph{$\Del$-set} if there is a coterm $t\in\Del A$
such that $B = \set{s^A}{s\ll t}$, where $s^A$ is the interpretation of the partial term $s$ in $A$.
As the evaluation map $s\mapsto s^A$ is monotone,
$B$ is a countable directed subset of $A$, therefore its down-closure $\down B$ is an $\omega$-ideal
of $A$. A \emph{$\Delta$-ideal} is the down-closure of a nonempty $\Del$-set $B$.
We say that an ordered algebra $A$ is \emph{$\Del$-regular} if
the suprema of all $\Del$-sets exist and the algebraic operations preserve
the suprema of nonempty $\Del$-sets.
\end{defi}

Note that $\Del X$ itself is a $\Del$-regular algebra due to the fact that $\Del$ is a monad. We will show that any $\Del$-regular algebra can be extended to an Eilenberg-Moore algebra for the monad $\Del$.
This means that the evaluation map $\eps_A:\FT A\to A$, $\eps_A(t)=t^A$ can be extended to domain $\Del A$ so that the following properties are satisfied:
\begin{gather}%
\label{eq:EM}
\begin{array}{c}
\begin{tikzpicture}[->, >=stealth', node distance=20mm, auto]
\small
 \node (NW) {$\Del^2A$};
 \node (NE) [right of=NW] {$\Del A$};
 \node (SW) [below of=NW, node distance=12mm] {$\Del A$};
 \node (SE) [right of=SW] {$A$};
 \path (NW) edge node {$\Del\eps_A$} (NE)
      edge node[swap] {$\mu_A$} (SW);
 \path (SW) edge[swap] node {$\eps_A$} (SE);
 \path (NE) edge node {$\eps_A$} (SE);
 \node (NW) [right of=NE] {$A$};
 \node (NE) [right of=NW] {$\Del A$};
 \node (SE) [below of=NE, node distance=12mm] {$A$};
 \path (NW) edge node {$\eta_A$} (NE)
      edge node[swap, xshift=3pt, yshift=3pt] {$\id_A$} (SE);
 \path (NE) edge node {$\eps_A$} (SE);
\end{tikzpicture}
\end{array}
\end{gather}
Moreover, there is a unique such extension making the evaluation map $\eps_A:\Del A\to A$ a $\Del$-continuous $\Sigma$-algebra morphism.

\begin{thm}%
\label{thm:Delalg1}
Let $A$ be a $\Del$-regular algebra. Extend the evaluation map to domain $\Del A$ by
\begin{align}
t^A = \bigvee\set{s^A}{s\ll t}.\label{eq:EMdef}
\end{align}
The supremum exists by the assumption that $A$ is $\Del$-regular. Then $A$ with the extended
evaluation map $\eps_A(t)=t^A$ for $t\in\Del A$ is an Eilenberg-Moore algebra for the monad $\Del$.
Moreover, $\eps_A:\Del A\to A$ is a $\Del$-continuous $\Sigma$-algebra morphism, and is the unique extension of $\eps_A:\FT A\to A$ for which this is so.
\end{thm}
\begin{proof}
We argue the latter statement first. The map $\eps_A$ is strict, as the term $\bot^A$ evaluates to itself. For monotonicity, let $s,t\in\Del A$, $s\le t$. By transitivity, $u\ll s$ implies $u\ll t$, thus
\begin{align*}
s^A = \bigvee\set{u^A}{u\ll s} \le \bigvee\set{u^A}{u\ll t} = t^A.
\end{align*}
For $\Del$-continuity, if $t\in\Del A$ then $\set s{s\ll t}$ is a $\Del$-set in $\Del A$ and $t=\bigvee\set s{s\ll t}$, thus
\begin{align*}
t^A = \bigvee\set{s^A}{s\ll t}\  &\Iff\ (\bigvee\set{s}{s\ll t})^A = \bigvee\set{s^A}{s\ll t}, 
\end{align*}
that is, the $\Del$-continuity condition for $t$ is exactly~\eqref{eq:EMdef}. Thus the extension is $\Del$-continuous and is unique.

Now we argue that the Eilenberg-Moore properties~\eqref{eq:EM} hold. The right-hand property holds since $\eta_A(a)$ is the term $a$.

For the left-hand property, we wish to show that $\Del{\eps_A(t)}^A={\mu_A(t)}^A$ for $t\in\Del^2 A$. For any such $t$, there exist $u\in\Del X$ and $h:X\to\Del A$ such that $t=\Del h(u)$, the coterm obtained by simultaneously substituting $h(x)$ for $x$ in $u$ for all $x\in X$.

Let us write $h'\ll h$ if $h'(x)\ll h(x)$ for all $x$. If $G$ is a collection of maps $g:X\to A$ such that $\set{g(x)}{g\in G}$ has a supremum for all $x$, define $(\bigvee G)(x)=\bigvee\set{g(x)}{g\in G}$. By~\eqref{eq:EMdef},
\begin{align*}
\eps_A\circ h &= \bigvee\set{\eps_A\circ h'}{h'\ll h}.
\end{align*}
It follows that for $v\in\FT X$,
\begin{align*}
{\FT(\eps_A\circ h)(v)}^A
&= \FT(\bigvee\set{\eps_A\circ h'}{h'\ll h})(v)^A = \bigvee\set{{\FT(\eps_A\circ h')(v)}^A}{h'\ll h} 
\end{align*}
since the algebraic operations of $A$ preserve suprema of nonempty $\Del$-sets. Also,
\begin{gather*}
\Del\eps_A(t) = \Del\eps_A(\Del h(u)) = \Del(\eps_A\circ h)(u)\\
s\ll\Del\eps_A(t) \Iff s\ll\Del(\eps_A\circ h)(u) \Iff \exists v\ll u\ s=\FT(\eps_A\circ h)(v),
\end{gather*}
therefore
\begin{align}
\Del{\eps_A(t)}^A &= \bigvee\set{s^A}{s\ll \Del\eps_A(t)} = \bigvee\set{{\FT(\eps_A\circ h)(v)}^A}{v\ll u}\nonumber\\
&= \bigvee\set{\bigvee\set{{\FT(\eps_A\circ h')(v)}^A}{h'\ll h}}{v\ll u}\nonumber\\
&= \bigvee\set{{\FT(\eps_A\circ h')(v)}^A}{h'\ll h,\ v\ll u}.\label{eq:EM4}
\end{align}
Similarly,
\begin{align*}
s\ll\mu_A(t) &\Iff s\ll\mu_A(\Del h(u)) \Iff \exists h'\ll h\ \exists v\ll u\ s=\mu_A(\FT h'(v)),
\end{align*}
therefore
\begin{align}
{\mu_A(t)}^A &= \bigvee\set{s^A}{s\ll\mu_A(t)} = \bigvee\set{{\mu_A(\FT h'(v))}^A}{h'\ll h,\ v\ll u}.\label{eq:EM5}
\end{align}
Thus~\eqref{eq:EM4} and~\eqref{eq:EM5} are equal provided
\begin{align*}
{\FT(\eps_A\circ h')(v)}^A &= {\mu_A(\FT h'(v))}^A
\end{align*}
for any $v\in\FT X$ and $h':X\to\FT A$. But this is just the left-hand diagram of~\eqref{eq:EM} for the functor $\FT$, which holds since $A$ is an Eilenberg-Moore algebra for $\FT$.
\end{proof}

\begin{defi}
A \emph{$\Del$-algebra} is an Eilenberg-Moore algebra for the monad $\Del$ on an ordered set $A$ with $\bot$ such that the evalution map $\eps_A:\Del A\to A$ is a $\Del$-continuous $\Sigma$-algebra morphism (preserves suprema of $\Del$-sets in $A$). A \emph{morphism} $h:A\to B$ of $\Del$-algebras is a strict monotone function that commutes with the evaluation maps $\eps_A:\Del A\to A$ and $\eps_B:\Del B\to B$.
\end{defi}

It follows that any morphism of $\Del$-algebras is $\Del$-continuous, as we now show.

\begin{thm}%
\label{thm:Delalg2}
Let $A$ and $B$ be $\Del$-algebras and $h:A\to B$ a $\Del$-algebra morphism. For any $\Del$-set $D\subs A$,
its image $\set{h(a)}{a\in D}\subs B$ is a $\Del$-set in $B$, and
\begin{align*}
h(\bigvee D) &= \bigvee \set{h(a)}{a\in D}.
\end{align*}
\end{thm}
\begin{proof}
Let $t\in\Del A$ such that $D = \set{s^A}{s\ll t}$. Then $t^A=\bigvee D$. We must show that $\set{h(s^A)}{s\ll t}$ is a $\Del$-set in $B$ and
\begin{align*}
h(t^A) &= \bigvee \set{h(s^A)}{s\ll t}.
\end{align*}
Since $h$ commutes with the evaluation maps, $h(s^A)=\Del {h(s)}^B$ for any $s\in\FT A$; thus
\begin{align*}
\set{h(s^A)}{s\ll t}
&= \set{\Del {h(s)}^B}{s\ll t}
= \set{u^B}{u\ll \Del h(t)},
\end{align*}
and this is a $\Del$-set in $B$. Moreover,
\begin{align*}
h(t^A)
= \Del {h(t)}^B
= \bigvee \set{u^B}{u\ll \Del h(t)}
= \bigvee \set{h(s^A)}{s\ll t}.
\tag*{\qedhere}
\end{align*}
\end{proof}

We have shown that every $\Del$-regular $\Sigma$-algebra extends canonically to a $\Del$-algebra. Conversely, every $\Del$-algebra is $\Del$-regular: suprema of $\Del$-sets exist, as $\bigvee(\set{s^A}{s\ll t})=t^A$ for $t\in\Del A$, and the algebraic operations are $\Del$-continuous, since for $\seq t1n\in\Del A$,
\begin{align*}
\lefteqn{f^A(\bigvee\set{s_1^A}{s_1\ll t_1},\ldots,\bigvee\set{s_n^A}{s_n\ll t_n})}\quad\\
&= f^A(t_1^A,\ldots,t_n^A) = {f(\seq t1n)}^A
= \bigvee\set{s^A}{s\ll f(\seq t1n)}\\
&= \bigvee\set{{f(\seq s1n)}^A}{s_i\ll t_i,\ 1\le i\le n}
= \bigvee\set{f^A(\seq{s^A}1n)}{s_i\ll t_i,\ 1\le i\le n}.
\end{align*}
Thus the categories of $\Del$-algebras and $\Del$-regular algebras with strict $\Del$-continuous morphisms are equivalent. Henceforth, we drop the ``regular'' and just call them $\Del$-algebras.

\begin{exa}
An $\FT$-algebra is just an ordered $\Sigma$-algebra.
If $\Del X$ is the collection of all regular coterms over $X$ (those with
only finitely many subterms up to isomorphism),
then a $\Del$-algebra is an ordered regular algebra~\cite{Tiuryn78,Tiuryn79}.
If $\Del X$ is the set of all regular coterms of order $n$,
$n\geq 0$, then a $\Del$-algebra is an $n$-regular (or $n$-rational) algebra~\cite{Gallier81}.
\end{exa}

We have defined $\Del$-algebras $A$ in terms of suprema of $\Del$-sets in $A$.
However, one could generalize the notion of $\Del$-set to include
any set of terms with a supremum in $\Del A$.

Suppose that $A$ is a $\Del$-algebra. A subset $E\subs\Del A$ is
said to be \emph{consistent} if any two elements of $E$, considered
as labeled trees, agree wherever both are defined. That is, if $s,t\in E$
and $\alpha\in\dom s\cap\dom t$, then $s(\alpha)=t(\alpha)$. Any consistent set
has a unique supremum in $\CT A$.

We say that a set $D\subs A$ is an \emph{extended $\Del$-set}
if there is a consistent set $E\subs\Del A$ such that $\bigvee E\in\Del A$
and $D = \set{t^A}{t\in E}$.

We state the following theorem without proof, but it is not difficult to prove
using the same technique as Theorems~\ref{thm:Delalg1} and~\ref{thm:Delalg2}.
\begin{thm}%
\label{thm:Delalg3}
The algebraic operations of any $\Del$-algebra and all $\Del$-algebra morphisms
preserve suprema of nonempty extended $\Del$-sets.
\end{thm}

Every $\Del$-set is an extended $\Del$-set, so we can replace the definition
of $\Del$-regular algebra by the stronger property that suprema of all nonempty extended
$\Del$-sets exist and are preserved by the algebraic operations.

\section{Main Result}

We can introduce varieties of $\Del$-algebras defined
by sets of inequalities between finite partial terms in $\FT X_\omega$
in the expected way.

Let $\cVo$, $\cVd$, and $\cV$ be the varieties of all $\omega$-continuous $\Sigma$-algebras,
$\Del$-algebras, and ordered $\Sigma$-algebras,
respectively, defined by a set of inequalities between finite partial terms.
Let $\FoX$ denote the free $\omega$-continuous
$\Sigma$-algebra on generators $X$ in $\cVo$ (see~\cite{Guessarian81}).
Extend $\FoX$ to a $\CT$-algebra as in Theorem~\ref{thm:Delalg1}, and
let $\FdX$ denote the $\Del$-subalgebra of $\FoX$ generated by $X$. The elements of $\FdX$ are the elements of $\FoX$ denoted by the coterms in $\Del X$:
\begin{align*}
\FdX &= \set{t^{\FoX}}{t\in\Del X},
\end{align*}
where $t^{\FoX}$ is the value of $t$ in $\FoX$.
Let $\FX$ denote the free ordered $\Sigma$-algebra on generators $X$ in $\cV$.

\begin{thm}%
\label{thm:main}
$\FdX$ is the free $\Del$-algebra in $\cVd$ on generators $X$. It is isomorphic
to the completion of $\FX$ by $\Del$-ideals.
\end{thm}
\begin{proof}
Since $\FdX$ embeds in $\FoX$ and $\FoX\in\cVo\subs\cVd$, we have that $\FdX\in \cVd$ by Birkhoff's theorem for ordered algebras~\cite{Bloom76}.

Let $A\in\cVd$ and $h:X \to A$.
Consider the completion $\Io A$ of $A$ by $\omega$-ideals,
which is an $\omega$-continuous algebra in $\cVo$~\cite{Bloom76,Guessarian81}.
The function $a\mapsto\down{\{a\}}$ embeds $A$ in $\Io A$,
and implicitly composing with this map allows us to view $h$ as a function $X\to \Io A$.
Since $\Io A \in \cVo$, there is a unique extension
of $h$ to a morphism $\hat h:\FoX \to \Io A$ of $\omega$-continuous algebras.
The restriction of $\hat h$ to $\FdX$ is a $\Del$-algebra
morphism $\FdX \to \Io A$. Let $\hat A$ denote the image of $\FdX$ under
$\hat h$. Then $\hat A$ is a $\Del$-subalgebra
of $\Io A$, and $\hat A$ is in $\cVd$ because it is a homomorphic image of $\FdX$, which is in $\cVd$.

Finally, let $\bigvee:\hat A\to A$ be the supremum operator. The elements
of $\hat A$ are $\Del$-ideals, therefore their suprema
exist and lie in $A$. The map $\bigvee$, being a component of the counit of the completion construction,
is a morphism of $\Del$-algebras; that is, a strict monotone function that preserves
suprema of $\Del$-sets and commutes with the algebraic operations.
The restriction of $\hat h$ composed with $\bigvee$ is the required unique extension
of $h$ to a morphism $\FdX\to A$ of $\Del$-algebras.

We can also construct $\FdX$ directly as the free extension of $\FX$ by $\Del$-ideals.
Here we use the general construction of \S\ref{sec:completion}. Let $\IDel$ be the monad of $\Del$-ideals in $\FX$.
We must check that $\IDel$ satisfies the condition~\eqref{eq:continuous}. For
$\Del$-ideals $\set{s^A}{s\ll t_i}\in\IDel\FX$ defined by coterms $\seq t1n\in\Del\FX$, $1\le i\le n$,
\begin{align*}
&\down{\set{f^A(\seq{s^A}1n)}{s_i^A\in\set{s^A}{s\ll t_i},\ 1\le i\le n}} \\
&\quad= \down{\set{{f(\seq s1n)}^A}{s_i\ll t_i,\ 1\le i\le n}}\\
&\quad= \down{\set{u^A}{u\ll f(\seq t1n)}} \in \IDel\FX.
\end{align*}
This allows algebraic operations on $\IDel\FX$ to be defined as in~\eqref{eq:continuous} and extended as in Theorem~\ref{thm:Delalg1} to give a $\Del$-algebra. Moreover, again by Birkhoff's theorem for ordered algebras, $\IDel\FX\in\cVd$ because it is a subalgebra of $\Io\FX\in\cVo\subs\cVd$~\cite{Bloom76,Guessarian81}.

We now argue that $\FdX$ and $\IDel\FX$ are isomorphic. Since $\IDel\FX\in\cVd$ and $\FdX$ is free for that variety, there is a unique $\Del$-algebra morphism $h:\FdX\to\IDel\FX$ extending the identity on $X$. Conversely, since $\FdX\in\cVd\subs\cV$ and $\FX$ is free in $\cV$, there is a unique $\Sigma$-algebra morphism $g:\FX\to\FdX$ extending the identity on $X$. By Theorem~\ref{thm:univ}, $g$ extends to a $\Del$-continuous morphism $\hat g:\IDel\FX\to\FdX$. The composition $\hat g\circ h:\FdX\to\FdX$
is the unique morphism on $\FdX$ extending the identity on $X$, therefore must be the identity morphism. Likewise, by Theorem~\ref{thm:univ}, the composition $h\circ\hat g:\IDel\FX\to\IDel\FX$ is the unique $\Del$-continuous morphism on $\IDel\FX$ extending the identity on $\FX$, therefore must be the identity morphism. Thus $h$ and $\hat g$ are inverses and $\FdX$ and $\IDel\FX$ are isomorphic.
\end{proof}

\section{Examples}%
\label{sec:examples}

\begin{exa}
The free $\omega$-continuous semiring on generators
$X$ is the semiring of power series $\naturals_\infty\angles{X\star}$
with countable support
having coefficients in $\naturals_\infty$, where $\naturals_\infty$ is obtained from the semiring of natural numbers
by adding a point at infinity; equivalently, the family of finite and countably infinite
multisets of strings in $X\star$.
The free regular semiring over $X$ is the semiring
$\naturals_\infty^{\rm alg}\angles{X\star}$ of all algebraic elements of $\naturals_\infty\angles{X\star}$,
and the free linear semiring is the semiring
$\naturals_\infty^{\rm rat}\angles{X\star}$ of all rational elements of $\naturals_\infty\angles{X\star}$.

In these examples, the signature is $+, \cdot, 0, 1$.
The free $\omega$-continuous semiring $\naturals_\infty\angles{X\star}$ is isomorphic to the
$\omega$-algebra of coterms $\CT X$ over this signature modulo the semiring axioms
and contains all suprema defined by coterms in $\CT X$.
Theorem~\ref{thm:main} says that the algebraic elements
$\naturals_\infty^{\rm alg}\angles{X\star}$ are
the images in $\naturals_\infty\angles{X\star}$ of the regular coterms $\RT X$, those
with only finitely many distinct subterms up to isomorphism.
The rational elements
$\naturals_\infty^{\rm rat}\angles{X\star}$ are the images of the linear coterms $\LT X$, those
defined as unrollings of finite systems of linear affine equations of the form
\begin{align}
x = a_1y_1 + \cdots + a_{n}y_n + b.\label{eq:starcont}
\end{align}
\end{exa}

\begin{exa}%
\label{ex:ex2}
The free $\omega$-continuous idempotent semiring on $X$ is the
semiring of all finite and countable languages in $X\star$.
The $\omega$-continuous idempotent semirings are the same as the
closed semirings used in the study of shortest-path algorithms~\cite{AHU75}.
Here the signature is again $+,\cdot,0,1$, although by the construction
of Theorem~\ref{thm:Delalg1}, one can adjoin a countable
summation operator $\sum$ that takes the supremum
of any countable set in the natural order $x\le y$ iff $x+y=y$.
\end{exa}

\begin{exa}%
\label{ex:ex3}
The free star-continuous Kleene algebra on generators $X$
is the algebra of regular subsets of $X\star$~\cite{K90a}. These are the
rational elements of the previous example. This is the free $\LT$-algebra for $\LT$ the set of
linear coterms described in~\eqref{eq:starcont}. The signature has been augmented with a
star operation $a\star$ that gives the least solution of the equation $x = 1+ax$.

The completeness and continuity conditions imposed by~\eqref{eq:starcont} can be succinctly axiomatized by saying that for every $a,b,c\in K$, $ab\star c$ is the supremum of $\set{ab^{n}c}{n\geq 0}$ in the natural order of the idempotent semiring. A Kleene algebra is called \emph{star-continuous} when it satisfies this condition. This is a special case of~\eqref{eq:starcont} corresponding to the system
\begin{align*}
x &= ay & y &= by + c
\end{align*}
which generates the regular coterm obtained by unwinding
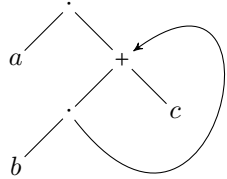
\begin{gather*}
\begin{tikzpicture}[-, >=stealth', node distance=10mm, inner sep=1pt, auto]
\small
 \useasboundingbox (0,-2.2) rectangle (0,.2);
 \node (root) {$\cdot$};
 \node (0) [below left of=root] {$a$};
 \node (1) [below right of=root] {$+$};
 \node (10) [below left of=1] {$\cdot$};
 \node (11) [below right of=1] {$c$};
 \node (100) [below left of=10] {$b$};
 \path (root) edge (0) edge (1);
 \path (1) edge (10) edge (11);
 \path (10) edge (100);
 \draw[->] (10) .. controls (2,-4) and (3,1) .. (1);
\end{tikzpicture}
\end{gather*}

The free $\omega$-continuous idempotent semiring of Example~\ref{ex:ex2} is the completion of the free star-continuous
Kleene algebra by countably generated star-ideals~\cite{K90a}.
\end{exa}

\begin{exa}%
\label{ex:ex4}
As first observed by Mezei and Wright~\cite{MezeiWright67}, the context-free languages
are the least solutions of finite systems of algebraic fixpoint equations of the form
\begin{align}
x = f(\seq y1n)\label{eq:algebraic}
\end{align}
in the idempotent semiring of languages. Such systems are essentially context-free grammars.
For example, the context-free language $\set{a^{n}b^n}{n\geq 0}$
is the least solution of the equation $x = axb + 1$ in this algebra, or in grammar form, $S \to aSb \mid \eps$.

The family of context-free languages
is the free $\RT$-algebra for $\RT X$ the set of regular coterms over $X$ defined by finite systems
of equations of the form~\eqref{eq:algebraic}
in the variety of idempotent semirings. For example, the fact that a least solution of $x = axb+1$
exists is the completeness condition expressed by the regular coterm obtained by unwinding
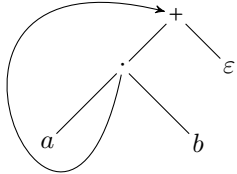
\begin{gather*}
\begin{tikzpicture}[-, >=stealth', node distance=10mm, inner sep=1pt, auto]
\small
 \useasboundingbox (0,-2.2) rectangle (0,.2);
 \node (root) {$+$};
 \node (0) [below left of=root] {$\cdot$};
 \node (1) [below right of=root] {$\eps$};
 \node (01) [below of=0] {};
 \node (00) [left of=01] {$a$};
 \node (02) [right of=01] {$b$};
 \path (root) edge (0) edge (1);
 \path (0) edge (00) edge (02);
 \draw[->] (0) .. controls (-1.5,-4.5) and (-4,1) .. (root);
\end{tikzpicture}
\end{gather*}
This coterm must have an interpretation in the algebra and must be the supremum of all its finite truncations.
Furthermore, by composing a coterm on the left and right with arbitrary $c$ and $d$, the resulting
coterm must also be the supremum of its truncations; this expresses the continuity of multiplication
with respect to definable suprema.

As in Example~\ref{ex:ex3}, the completeness and continuity conditions imposed by~\eqref{eq:algebraic}
can be axiomatized by a special case, namely the \emph{$\mu$-continuity axiom}, which says that $a(\mu x.b)c$ is the supremum of $\set{a(nx.b)c}{n\ge 0}$, where $0x.b = \bot$, $(n+1)x.b = b[x/nx.b]$, and $\mu x.b$ is the least solution of the equation $x = b$~\cite{GHK17a}.
\end{exa}

\begin{exa}
The free ``1-regular'' idempotent semiring over an alphabet $X$ is the semiring
of OI-macro languages of Fischer~\cite{Fischer68}, which are the same as the
indexed languages of Aho~\cite{Aho68,AHU75}. These languages are quite useful in
practice, as they capture some syntactic properties that are inexpressible
with context-free grammars; for example, that a variable may only appear
in the scope of its declaration~\cite{ThiemannNeubauer08}. Macro grammars allow these
properties to be specified syntactically and checked during the parsing phase
of compilation.

Indexed languages are generated by \emph{indexed grammars}
in which the nonterminals are tagged with a stack. Productions are of the form
\begin{align*}
& A[\sigma]\to\alpha[\sigma]
&& A[\sigma]\to\alpha[F\sigma]
&& A[F\sigma]\to\alpha[\sigma]
\end{align*}
where $A$ is a nonterminal symbol, $\sigma$ is a stack (string of stack symbols), $F$ is a stack symbol, and $\alpha$ is a string of
terminals and nonterminals. The notation $\alpha[\sigma]$ denotes the string of terminals and tagged nonterminals
in which every nonterminal in $\alpha$ is tagged with $\sigma$. A string $x$ of terminal
symbols is in the language generated by the grammar provided $S[]\derives x$, where $S$ is a designated start symbol.

Indexed grammars can generate non-context-free languages. For example, here are grammars for the sets
$\set{a^{n}b^{n}c^{n}}{n\geq 0}$ and $\set{a^{2^n}}{n\geq 0}$, respectively:
\begin{align*}
& S[\sigma]\to S[F\sigma] \mid A[\sigma]B[\sigma]C[\sigma] && S[\sigma]\to S[F\sigma] \mid T[\sigma]\\
& A[F\sigma]\to aA[\sigma]\qquad A[]\to \eps        && T[F\sigma]\to T[\sigma]T[\sigma]\\
& B[F\sigma]\to bB[\sigma]\qquad B[]\to \eps        && T[]\to a\\
& C[F\sigma]\to cC[\sigma]\qquad C[]\to \eps
\end{align*}
The same class of languages are generated by \emph{OI-macro grammars}~\cite{Fischer68}. These are grammars whose nonterminals may have parameters. For example, the languages $\set{a^{n}b^{n}c^n}{n\geq 0}$ and $\set{a^{2^n}}{n\geq 0}$ are generated by the following two OI-macro grammars, respectively:
\begin{align*}
& S\to F(\eps,\eps,\eps)               && S\to F(a)\\
& F(x,y,z)\to F(ax,by,cz) \mid xyz          && F(x) \to F(xx) \mid x
\end{align*}

As with context-free grammars, indexed grammars and OI-macro grammars can be viewed as schemes for generating a quasi-regular family of coterm algebras $\Delta X$. For example, the coterms generated by the grammars above are, respectively,
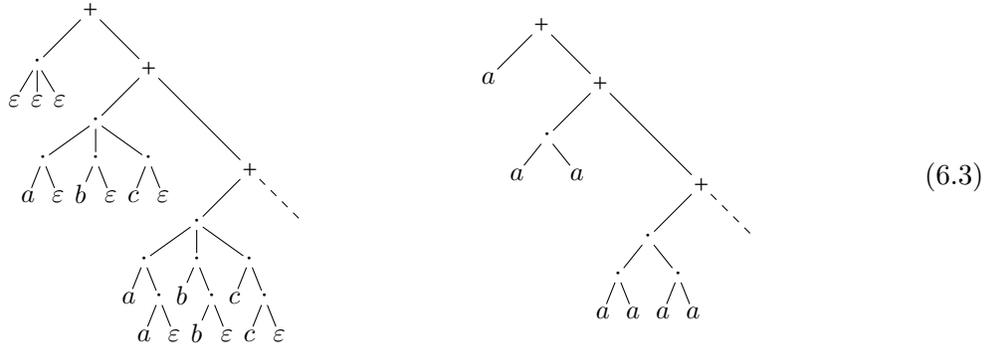
\begin{gather}%
\label{eq:1regular}
\begin{array}{c@{\hspace{1cm}}c}
\begin{tikzpicture}[-, >=stealth', node distance=10mm, inner sep=1pt, auto]
 \small
 \useasboundingbox (-1.1,-4.4) rectangle (8.8,.1);
 \node (root) {$+$};
 \node (0) [below left of=root] {$\cdot$};
 \node (1) [below right of=root, node distance=11mm] {$+$};
 \node (01) [below of=0, node distance=5mm] {$\eps$};
 \node (00) [left of=01, node distance=3mm] {$\eps$};
 \node (02) [right of=01, node distance=3mm] {$\eps$};
 \node (10) [below left of=1] {$\cdot$};
 \node (101) [below of=10, node distance=5mm] {$\cdot$};
 \node (100) [left of=101, node distance=7mm] {$\cdot$};
 \node (102) [right of=101, node distance=7mm] {$\cdot$};
 \node (11) [below right of=1, node distance=19mm] {$+$};
 \node (110) [below left of=11] {$\cdot$};
 \node (111) [below right of=11] {};
 \node (1001) [below of=100, node distance=5mm] {};
 \node (1000) [left of=1001, node distance=2mm] {$a$};
 \node (1002) [right of=1001, node distance=2mm] {$\eps$};
 \node (1011) [below of=101, node distance=5mm] {};
 \node (1010) [left of=1011, node distance=2mm, yshift=1pt] {$b$};
 \node (1012) [right of=1011, node distance=2mm] {$\eps$};
 \node (1021) [below of=102, node distance=5mm] {};
 \node (1020) [left of=1021, node distance=2mm] {$c$};
 \node (1022) [right of=1021, node distance=2mm] {$\eps$};
 \node (1101) [below of=110, node distance=5mm] {$\cdot$};
 \node (1100) [left of=1101, node distance=7mm] {$\cdot$};
 \node (1102) [right of=1101, node distance=7mm] {$\cdot$};
 \node (11001) [below of=1100, node distance=5mm] {};
 \node (11000) [left of=11001, node distance=2mm] {$a$};
 \node (11002) [right of=11001, node distance=2mm] {$\cdot$};
 \node (11011) [below of=1101, node distance=5mm] {};
 \node (11010) [left of=11011, node distance=2mm, yshift=1pt] {$b$};
 \node (11012) [right of=11011, node distance=2mm] {$\cdot$};
 \node (11021) [below of=1102, node distance=5mm] {};
 \node (11020) [left of=11021, node distance=2mm] {$c$};
 \node (11022) [right of=11021, node distance=2mm] {$\cdot$};
 \node (110021) [below of=11002, node distance=5mm] {};
 \node (110020) [left of=110021, node distance=2mm] {$a$};
 \node (110022) [right of=110021, node distance=2mm] {$\eps$};
 \node (110121) [below of=11012, node distance=5mm] {};
 \node (110120) [left of=110121, node distance=2mm, yshift=1pt] {$b$};
 \node (110122) [right of=110121, node distance=2mm] {$\eps$};
 \node (110221) [below of=11022, node distance=5mm] {};
 \node (110220) [left of=110221, node distance=2mm] {$c$};
 \node (110222) [right of=110221, node distance=2mm] {$\eps$};
 \path (root) edge (0) edge (1);
 \path (0) edge (01) edge (02) edge (00);
 \path (1) edge (10) edge (11);
 \path (10) edge (100) edge (101) edge (102);
 \path (100) edge (1000) edge (1002);
 \path (101) edge (1010) edge (1012);
 \path (102) edge (1020) edge (1022);
 \path (11) edge (110) edge[dashed] (111);
 \path (110) edge (1100) edge (1101) edge (1102);
 \path (1100) edge (11000) edge (11002);
 \path (1101) edge (11010) edge (11012);
 \path (1102) edge (11020) edge (11022);
 \path (11002) edge (110020) edge (110022);
 \path (11012) edge (110120) edge (110122);
 \path (11022) edge (110220) edge (110222);
 \node (rootA) [right of=root, node distance=60mm, yshift=-2mm] {$+$};
 \node (0) [below left of=rootA] {$a$};
 \node (1) [below right of=rootA, node distance=11mm] {$+$};
 \node (10) [below left of=1] {$\cdot$};
 \node (101) [below of=10, node distance=5mm] {};
 \node (100) [left of=101, node distance=4mm] {$a$};
 \node (102) [right of=101, node distance=4mm] {$a$};
 \node (11) [below right of=1, node distance=19mm] {$+$};
 \node (110) [below left of=11] {$\cdot$};
 \node (111) [below right of=11] {};
 \node (1101) [below of=110, node distance=5mm] {};
 \node (1100) [left of=1101, node distance=4mm] {$\cdot$};
 \node (1102) [right of=1101, node distance=4mm] {$\cdot$};
 \node (11001) [below of=1100, node distance=5mm] {};
 \node (11000) [left of=11001, node distance=2mm] {$a$};
 \node (11002) [right of=11001, node distance=2mm] {$a$};
 \node (11021) [below of=1102, node distance=5mm] {};
 \node (11020) [left of=11021, node distance=2mm] {$a$};
 \node (11022) [right of=11021, node distance=2mm] {$a$};
 \node (dummy) [below of=11021, node distance=7mm] {};
 \path (rootA) edge (0) edge (1);
 \path (1) edge (10) edge (11);
 \path (10) edge (100) edge (102);
 \path (11) edge (110) edge[dashed] (111);
 \path (110) edge (1100) edge (1102);
 \path (1100) edge (11000) edge (11002);
 \path (1102) edge (11020) edge (11022);
\end{tikzpicture}
\end{array}
\end{gather}
These are non-regular trees representing the elements
$\sum_{n=0}^\infty a^{n}b^{n}c^n$ and
$\sum_{n=0}^\infty a^{2^n}$ of a $\Delta$-regular algebra, respectively.

A context-free language is the frontier of a parse tree whose paths are elements of a regular language.
Similarly, an indexed language is the frontier of a parse tree whose paths are elements of a context-free language.
Iterating this idea, one arrives at a hierachy of language families, which have been studied as a model
of higher-order recursion and in linguistics~\cite{KobeleSalvati15,Damm82,DammGoerdt86,EngelfrietSchmidt77,EngelfrietSchmidt78,Courcelle78a,Courcelle78b,BensonGuessarian87,Gallier81,Gallier85,Tiuryn78,Tiuryn79}. These families are defined in terms of finite systems of higher-order fixpoint equations or nested stack automata. Such representations also give rise to syntactically restricted completeness and continuity properties as expressed by suitable submonads of $\CT$.
\end{exa}

\begin{exa}%
\label{ex:iteration}
A final example is given by the category of (\emph{ordered}) \emph{iteration theories}~\cite{BloomEsik93},
ordered Lawvere theories equipped with a fixpoint operator $(^\dagger)$ that provides
solutions to parameterized fixpoint equations. Iteration theories are a general
class of models that capture the equational properties
of structures that arise in domain theory.

It is shown in~\cite{BloomEsik93} that the rational partial coterms over a signature $\Gamma$
form an iteration theory, and in fact
the free ordered iteration theory on generators $\Gamma$. It is further
shown in~\cite{AdamekMiliusVelebil07} that the iteration theories
are precisely the Eilenberg-Moore algebras for the rational coterm monad
on the category of signatures. Let us briefly explain these
results and their significance in the current context.

A \emph{Lawvere theory}, as presented in~\cite{BloomEsik93},
is a category with objects $\naturals$ and distinguished associative coproduct
$+$ such that $n=1+\cdots+1$ ($n$ times). The coprojections are denoted $\inj ni$, $1\le i\le n$,
and the universal coproduct arrow is denoted $[\seq f1n]:n\to m$, where $f_i:1\to m$, $1\le i\le n$.
Lawvere theories are an abstraction of conventional algebraic varieties.
If $\Gamma$ is a ranked alphabet and $f\in\Gamma_n$, then for any $\Gamma$-algebra $A$,
the algebraic operation $f^A:A^n\to A$ is a morphism $f:1\to n$ in the corresponding Lawvere
theory. A morphism $g:n\to m$ models an $n$-tuple of $m$-ary functions. Lawvere theories are
traditionally defined covariantly in terms of products instead of coproducts, but~\cite{BloomEsik93} prefers the contravariant presentation for technical convenience.

An \emph{iteration theory} is a Lawvere theory with an added fixpoint operator $^\dagger$ satisfying
certain equations. The operator $^\dagger$ can be applied to any morphism $e:n\to n+p$
and yields a morphism $e^\dagger:n\to p$. This models the domain-theoretic
construction of taking the least fixpoint of a parameterized map
\begin{align*}
\lam{\bar x,\bar y}{e^A(\bar x,\bar y)}:A^n\times A^p\to A^n
\end{align*}
with respect to induction variables $\bar x=\seq x1n$
and parameters $\bar y = \seq y1p$ to give
\begin{align*}
(e^\dagger)^A &= \lam{\bar y}{\mathsf{fix}\,\bar x\kern1pt.\kern1pt e^A(\bar x,\bar y)}:A^p\to A^n.
\end{align*}

Iteration theories must satisfy a certain infinite set of equational axioms that
essentially characterize when two finite expressions with
composition, cotupling, and $^\dagger$ unwind to the same regular coterm with
composition and cotupling but without $^\dagger$.
A key axiom is
\begin{gather*}
\begin{tikzpicture}[->, >=stealth', auto]
\small
 \useasboundingbox (0,-.7) rectangle (0,.9);
 \node (O) {};
 \node (NW) at (120:.9) {$n$};
 \node (SW) at (240:.9) {$n+p$};
 \node (NE) at (0:.9) {$p$};
 \path (NW) edge node[swap] {$e$} (SW) edge node {$e^\dagger$} (NE);
 \path (SW) edge node[swap] {$[e^\dagger,\id_p]$} (NE);
\end{tikzpicture}
\end{gather*}
which asserts that $e^\dagger$ is indeed a fixpoint. Other axioms handle the interaction of $^\dagger$ with the other
operators and various structural properties.
\end{exa}

Formally, the category of signatures $\Sgn$ is the category of ranked alphabets and rank-preserving functions, or the slice category $\Set/\naturals$ as described in~\cite{AdamekMiliusVelebil07}.
The rational partial coterms $\RT_{\Gamma} X_\omega$, where $X_\omega=\{x_0,x_1,\ldots\}$, form an ordered iteration theory in which
the homsets $n\to m$ are $n$-tuples of coterms with variables
among $\seq x0{m-1}$ (recall that rational coterms have only finitely many isomorphic
subterms, so only finitely many variables). The empty coterm can be used polymorphically at any type $\bot_{nm}:n\to m$.
Composition is substitution and $^\dagger$ gives the limit of the
$k$-fold composition of a coterm with itself for $k=0,1,2,\ldots$~. This
is the free iteration theory on generators $\Gamma$, which says that
the forgetful functor that discards all morphisms but those of type $1\to n$
has a left adjoint. The composition of this left adjoint with the forgetful functor is the monad $\Rat:\Sgn\to\Sgn$ that takes a ranked alphabet $\Gamma$ to its rational partial coterms $\RT_{\Gamma} X_\omega$.

The most general category with the same monad is the category $\RatAlg$ of $\Rat$-algebras.
It is shown in~\cite{AdamekMiliusVelebil07} that these are precisely the iteration theories.
The restriction to rational coterms, as opposed to all coterms, serves to
limit the domain to those objects that have finite representations and are thus
computationally meaningful.

To explain how these results fit in the current context, we must first be careful about the use of the word ``signature'' because there are two different notions of signature in play. As used above, $\Gamma$ represents a conventional ranked alphabet $\Gamma = |\Gamma|\to\naturals$ for single-sorted algebras, an object of the category $\Sgn=\Set/\naturals$. These correspond to the sets $X$ in Theorem 5.1.

The other notion of signature is the signature by which we form partial coterms. This corresponds to the $\Sigma$ in the sense of $\SAlg$ as used throughout this paper. For Lawvere theories, the operations are cotupling, coprojection, and composition. For iteration theories, we add $\dagger$. These are actually typed signatures with infinitely many operations at different types, along with a set of typing rules for term formation. The types are $n\to m$, the types of morphisms of iteration theories. The typing rules are analogous to those for matrix operations on non-square matrices or the operations of typed Kleene algebra~\cite{K98b}.

We can form partial coterms with the operations of Lawvere theories, with cotupling and composition at the internal nodes and elements of $\Gamma$, coprojections, and the empty term acting as $\bot$ at the leaves. These partial coterms form $\CT(\Gamma)$. The regular (or rational) partial coterms form a submonad $\Del$ of $\CT$, thus $\Del(\Gamma)$ are the regular partial coterms over $\Gamma$. Modulo the axioms of Lawvere theories, these form the free (ordered) iteration theory on generators $\Gamma$, which is $F_\Del(\Gamma)$. The operation $\dagger$ is defined as the limit of $k$-fold compositions of a coterm with itself as described above. Theorem 5.1 says that this algebra is isomorphic to the completion of $F(\Gamma)$ by regular ideals, where $F(\Gamma)$ is the algebra of finite partial terms modulo the axioms of Lawvere theories, the free algebra in $\cV$ on generators $\Gamma$.

The notion of typed signature is not technically covered by the definition of signature in \S2, so for this example we do not work with the monads $\FT$ and $\CT$ on $\Set$, but rather with their well-typed submonads on $\Sgn$. The category $\Sgn$ is a concrete category, meaning that there is a faithful functor to $\Set$, thus $\Sgn$ can be regarded as a subcategory of $\Set$. Its objects $\Gamma$ are exactly those sets with the added typing information needed to allow the typed operations of cotupling, composition, and $\dagger$ to be applied in a well-typed way. The application of $\Del$ to such a set results in another signature in $\Sgn$, so the monad $\Del$ on $\Set$, restricted to $\Sgn$, lifts to a monad on $\Sgn$. This is exactly the monad $\Rat$.

\subsection*{Acknowledgments}
Thanks to the anonymous referees for their careful reading and valuable suggestions that greatly improved the presentation. This research was supported in part by NSF grant AiTF-1637532.

\bibliographystyle{plain}
\bibliography{dk,Regalg}

\end{document}